\documentclass[12pt]{article}

\usepackage{url}
\usepackage[matrix,arrow,curve]{xy}
\usepackage{amssymb}
\usepackage{amsmath}
\usepackage{amsthm}

\newtheorem{thm}{Theorem}
\newcommand{\cpsa}{\textsc{cpsa}}
\newcommand{\cn}[1]{\ensuremath{\operatorname{\mathsf{#1}}}}

\newcommand{\fn}[1]{\ensuremath{\operatorname{\mathit{#1}}}}
\newcommand{\sdom}{\fn{Dom}}
\newcommand{\sran}{\fn{Ran}}
\newcommand{\vars}{\fn{Vars}}
\newcommand{\seq}[1]{\ensuremath{\langle#1\rangle}}
\newcommand{\enc}[2]{\ensuremath{\{\!|#1|\!\}_{#2}}}
\newcommand{\inbnd}{\mathord -}
\newcommand{\outbnd}{\mathord +}
\newcommand{\srt}[1]{\ensuremath{\mathsf{#1}}}

\newcommand{\all}[1]{\mathop{\forall#1}}
\newcommand{\some}[1]{\mathop{\exists#1}}
\newcommand{\prefix}[2]{#1\dagger#2}
\newcommand{\init}{\fn{init}}
\newcommand{\resp}{\fn{resp}}
\newcommand{\form}{\mathcal{K}}
\newcommand{\sent}{\mathcal{S}}
\newcommand{\lang}{\mathcal{L}}
\newcommand{\alg}[1]{\ensuremath{\mathfrak#1}}
\newcommand{\alga}{\alg{A}}

\newcommand{\rl}{\fn{rl}}
\newcommand{\skel}{\mathsf{k}}
\newcommand{\insta}{\mathsf{i}}
\newcommand{\nodes}{\fn{nodes}}
\newcommand{\evt}{\fn{evt}}
\newcommand{\orig}{\mathcal{O}}

\title{Deducing Security Goals From \\ Shape Analysis Sentences}
\author{John D.\ Ramsdell}

\begin{document}
\maketitle

\begin{abstract}
Guttman presented a model-theoretic approach to establishing security
goals in the context of strand space theory.  In his approach, a run
of the Cryptographic Protocol Shapes Analyzer ({\cpsa}) produces
models that determine if a goal is satisfied.

This paper presents a method for extracting a sentence that completely
characterizes a run of {\cpsa}.  Logical deduction can then be used to
determine if a goal is satisfied.  This method has been implemented
and is available to all.
\end{abstract}

\section{Introduction}

\emph{This revision updates the strand-oriented protocol language
  described in Section~\ref{sec:sas} to one that has been shown to
  be useful in practice.  The November 2014 revision of the January
  2012 paper corrects several minor errors and adds a description of a
  node-oriented protocol language in
  Appendix~\ref{sec:node-oriented}.}

A central problem in cryptographic protocol analysis is to determine
whether a formula that expresses a security goal about behaviors
compatible with a protocol is true.  Following~\cite{Guttman09}, a
security goal is a quantified implication:
\begin{equation}\label{eqn:security goal}
\all{\vec{x}}(\Phi_0\supset\bigvee_{1\le i\le n}\some{\vec{y}_i}\Phi_i).
\end{equation}

The hypothesis~$\Phi_0$ is a conjunction of atomic formulas describing
regular (honest) behavior.  Each disjunct~$\Phi_i$ that makes up the conclusion
is also a conjunction of atomic formulas.  When~$\Phi_i$ describes
desired behaviors of other regular participants, then the formula is
an \emph{authentication} goal.  The goal says that each run of the
protocol compatible with~$\Phi_0$ will include the regular behavior
described by one of the disjuncts.  When $n=0$, the goal's conclusion
is false.  If~$\Phi_0$ mentions an unwanted disclosure,
(\ref{eqn:security goal}) says the disclosure cannot occur, thus a
security goal with $n=0$ expresses a \emph{secrecy} goal.

Guttman~\cite{Guttman09} presented a model-theoretic approach to
establishing security goals in the context of strand space theory.  In
that setting, a \emph{skeleton} describes regular behaviors compatible
with a protocol.  For skeleton~$k$ and formula~$\Phi$, he defined
$k,\alpha\models\Phi$ to mean that the conjunction of atomic formulas
that make up~$\Phi$ is satisfied in~$k$ with variable
assignment~$\alpha$.

A \emph{realized} skeleton is one that includes enough regular
behavior to specify all the non-adversarial part of an execution of
the protocol.  In a realized skeleton, its message transmissions
combined with possible adversarial behavior explain every message
reception in the skeleton.

In strand space theory, a \emph{homomorphism} is a
structure-preserving map~$\delta$ that shows how the behaviors in one
skeleton are reflected within another.  As skeletons serve as models,
homomorphisms preserve satisfaction for conjunctions of atomic
formulas.

The Cryptographic Protocol Shapes Analyzer ({\cpsa}) constructs
homomorphisms from a skeleton~$k_0$ to realized
skeletons~\cite{cpsa09}.  If {\cpsa} terminates, it generates a set of
realized skeletons~$k_i$ and a set of homomorphisms $\delta_i\colon
k_0\mapsto k_i$.  These realized skeletons are all the minimal,
essentially different skeletons that are homomorphic images of~$k_0$
and are called the \emph{shapes} of the analysis.

Guttman proposed a recipe for evaluating goal (\ref{eqn:security
  goal}) based on the following two technical results.
\begin{itemize}
\item For any security hypothesis~$\Phi_0$ there is a skeleton~$k_0$ that
  characterizes it in the sense that for all~$k$:
  $$\some{\alpha}k,\alpha\models
  \Phi_0\mbox{ iff }\some{\delta}\delta\colon k_0\mapsto k$$
\item There exists a realized skeleton that is a counterexample to
  (\ref{eqn:security goal}) iff there exists some shape in the
  analysis of~$k_0$ that is a counterexample.
\end{itemize}
These two results justify the following procedure.
\begin{enumerate}
\item Construct a characteristic skeleton~$k_0$ for~$\Phi_0$.
\item Ask {\cpsa} for the shapes produced by analyzing~$k_0$.
\item As {\cpsa} delivers shapes, check that each satisfies some
  disjunct~$\Phi_i$.
\item If the answer is no, this shape is a counterexample to
  (\ref{eqn:security goal}).
\item If {\cpsa} terminates with no counterexample, then
  (\ref{eqn:security goal}) is achieved.
\end{enumerate}

\paragraph{Implementing Security Goals.}

{\cpsa} now has support for security goals, but not as specified by
Guttman.  Part of the reason for the difference is that the details of
the formalism that underlies the {\cpsa}
implementation~\cite{cpsaspec09} dictate changes to the logic of
security goals.  These details will be elaborated later in this paper.

The key difference is a change in perspective.  Instead of finding a
formula that characterizes a security hypothesis, {\cpsa} includes a
tool that extracts a sentence that characterizes a shape analysis.
This so called \emph{shape analysis sentence} is special in that
it encodes everything that can be learned from the shape analysis.

Given a shape analysis sentence, a security goal is achieved if the
goal can be deduced from the sentence.  {\cpsa} includes a Prolog
program that translates shape analysis sentences into
Prover9~\cite{prover9} syntax.  Typically, a goal that is a theorem is
quickly proved by Prover9.

There is another advantage to this approach.  It can be tedious to
generate security goals.  Realistic ones can be large and complicated.
An easy way to create one is to modify a shape analysis sentence.
This typically involves deleting parts of the conclusion.

There is a disadvantage to this approach.  When a goal cannot be
deduced from a shape analysis sentence, one cannot conclude that there
is a counterexample.  It could be simply that the sentence is not
relevant to the security goal.

\paragraph{Motivating Example.}

The running example used throughout this paper is now presented.  An
informal version of the example is presented here, and the example
with all of the details filled is in Section~\ref{sec:example}.

The following protocol is a simplified version of the Denning-Sacco
key distribution protocol~\cite{DenningSacco81} due to Bruno
Blanchet~\cite{Blanchet11}.
$$\begin{array}{r@{{}:{}}l}
A\to B&\enc{\enc{s}{a^{-1}}}{b}\\
B\to A&\enc{d}{s}
\end{array}$$
Alice~($A$) freshly generates symmetric key~$s$, signs the symmetric
key with her private uncompromised asymmetric key $a^{-1}$ and intends
to encrypt it with Bob's~($B$) uncompromised asymmetric key $b$.
Alice expects to receive data~$d$ encrypted, such that only Alice and
Bob have access to it.

The protocol was constructed with a known flaw for expository
purposes, and as a result the secret is exposed due to an
authentication failure.  The protocol does not prevent Alice from
using a compromised key~$b'$, so that Mallory~($M$) can
perform this man-in-the-middle attack:
$$\begin{array}{r@{{}:{}}l}
A\to M&\enc{\enc{s}{a^{-1}}}{b'}\\
M\to B&\enc{\enc{s}{a^{-1}}}{b}\\
B\to E&\enc{d}{s}
\end{array}$$

The protocol fails to provide a means for Bob to ensure the original
message was encrypted using his key.  The authentication failure is
avoided with this variation of the protocol:
\begin{equation}\label{eq:amended protocol}
\begin{array}{r@{{}:{}}l}
A\to B&\enc{\enc{s,b}{a^{-1}}}{b}\\
B\to A&\enc{d}{s}
\end{array}
\end{equation}

In strand space theory, a \emph{strand} is a linearly ordered sequence
of events $e_0\Rightarrow\cdots\Rightarrow e_{n-1}$, and an
\emph{event} is either a message transmission $\bullet\to$ or a
reception $\bullet\gets$.  In {\cpsa}, adversarial behavior is not
explicitly represented, so strands always represent regular behavior.

Regular behavior is constrained by a set of roles that make up the
protocol.  In this protocol, Alice's behaviors must be compatible with
an initiator role, and Bob's behaviors follow a responder role.
\begin{equation}\label{eq:protocol}
\begin{array}{r@{\qquad}l}
\xymatrix@=.6em{\raisebox{-1ex}[0ex][0ex]{\strut\init}\\
  \bullet\ar@{=>}[d]\ar[r]&\enc{\enc{s}{a^{-1}}}{b}\\
  \bullet&\enc{d}{s}\ar[l]}&
\xymatrix@=.6em{&\raisebox{-1ex}[0ex][0ex]{\strut\resp}\\
\enc{\enc{s}{a^{-1}}}{b}\ar[r]&\bullet\ar@{=>}[d]\\
\enc{d}{s}&\bullet\ar[l]}
\end{array}
\end{equation}

The important authentication goal from Bob's perspective is that if an
instance of a responder role runs to completion, there must have been
an instance of the initiator role that transmitted its first message.
Furthermore, assuming the symmetric key is freshly generated, and the
private keys are uncompromised, the two strands agree on keys used for
signing and encryption.

A {\cpsa} analysis of the authentication goal requires two inputs, a
specification of the roles that make up the protocol, as in
Eq.~\ref{eq:protocol}, and a question about runs of the protocol.  The
question in this case is the hypothesis of Eq.~\ref{eq:shape
  analysis}, that an instance of the responder role ran to completion.
In these diagrams, a strand instantiated from a role is distinguished
from a role by placing messages above communication arrows, and
$\succ$ is used to assert an event occurred after another.

\begin{equation}\label{eq:shape analysis}
\begin{array}{rcl}
\xymatrix@R=1em@C=3em{\raisebox{-0.5ex}[0ex][0ex]{\strut\resp}\\
\bullet\ar@{=>}[d]&\ar[l]_{\enc{\enc{s}{a^{-1}}}{b}}\\
\bullet\ar[r]^{\enc{d}{s}}&}&
\raisebox{-4.0ex}{implies}&
\xymatrix@R=1em@C=3em{\raisebox{-0.5ex}[0ex][0ex]{\strut\resp}&&
\raisebox{-0.5ex}[0ex][0ex]{\strut\init}\\
\bullet\ar@{=>}[d]&\succ\ar[l]_{\enc{\enc{s}{a^{-1}}}{b}}&\bullet
\ar[l]_{\enc{\enc{s}{a^{-1}}}{b'}}\\
\bullet\ar[r]^{\enc{d}{s}}&}
\end{array}
\end{equation}

{\cpsa} produces the conclusion in Eq.~\ref{eq:shape analysis}, that
an instance of the initiator role must have transmitted its first
message, but it does not conclude that the strands agree on the key
used for the outer encryption.  When {\cpsa} is run using the amended
protocol in Eq.~\ref{eq:amended protocol}, the strands agree on the
key, and the authentication goal is achieved.

The contribution of this paper is a method of formalizing security
goals and the results of a {\cpsa} analysis in first-order logic such
that whenever a {\cpsa} analysis demonstrates that a security goal is
achieved, the logical sentence associated with the security goal will
be deducible from the shape analysis sentence with the relevant
{\cpsa} analysis.  The sentences associated with this example are
presented in Section~\ref{sec:example}.

\paragraph{Some Related Work.}

This paper is the result of implementing security goals as described
in Guttman~\cite{Guttman09}.  The original motivation for extracting
shape analysis sentences rather than following the procedure
in~\cite{Guttman09} was ease of implementation.  With shape analysis
sentences, most of the work is performed by a post-processing stage,
and there were only a few changes made to the core {\cpsa} program.
Only later it was realized the sense in which shape analysis sentences
completely characterize a shape analysis.

The Scyther tool~\cite{cremers06} integrates security goal
verification with its core protocol analysis algorithm.  Security
goals are easy to state as long as they can be expressed using a
predefined vocabulary, however, there is no sense in which Scyther
goals characterize an analysis.

The Protocol Composition Logic~\cite{datta05} provides a contrasting
approach to specifying security goals.  It extends strand spaces by
adding an operational semantics as a small set of reduction rules, and
a run of a protocol is a sequence of reduction steps derived from an
initial configuration.  The logic is a temporal logic interpreted over
runs.  The logic is more expressive than what is described within this
paper at the cost of added complexity.

\paragraph{Structure of this Paper.}

Section~\ref{sec:alg and morphs} describes the formalism on which
{\cpsa} is built, Section~\ref{sec:sas} presents the logic built upon
that formalism, and Section~\ref{sec:example} displays the example
above in full detail.

\paragraph{Notation.}

A finite sequence is a function from an initial segment of the natural
numbers.  The length of a sequence~$X$ is~$|X|$, and
sequence~$X=\seq{X(0),\ldots, X(n-1)}$ for $n=|X|$.  If~$S$ is a set,
then~$S^\ast$ is the set of finite sequences over~$S$, and~$S^+$ is the
non-empty finite sequences over~$S$.  The prefix of sequence~$X$ of
length~$n$ is~$\prefix{X}{n}$.

\iffalse
The concatenation of sequences~$X_0$ and~$X_1$ is~$X_0\append X_1$.
\fi

\section{Message Algebras and Homomorphisms}\label{sec:alg and morphs}

The two details of {\cpsa}'s formalism that dictate changes to the
logic of security goals are the fact that in {\cpsa}, a message
algebra is an order-sorted quotient term algebra and homomorphisms are
strand-oriented, not node-oriented.  The issues surrounding
homomorphisms will be described later.

\begin{figure}
$$\begin{array}{ll@{{}\colon{}}ll}
\mbox{Sorts:}&
\multicolumn{3}{l}{\mbox{$\top$, $\srt{A}$, $\srt{S}$, $\srt{D}$}}\\
\mbox{Subsorts:}&
\multicolumn{3}{l}{\mbox{$\srt{A}<\top$, $\srt{S}<\top$, $\srt{D}<\top$}}\\
\mbox{Operations:}&(\cdot,\cdot)&\top\times\top\to\top& \mbox{Pairing}\\
&\enc{\cdot}{(\cdot)}&\top\times\srt{A}\to\top&\mbox{Asymmetric encryption}\\
&\enc{\cdot}{(\cdot)}&\top\times\srt{S}\to\top&\mbox{Symmetric encryption}\\
&(\cdot)^{-1}&\srt{A}\to\srt{A}& \mbox{Asymmetric key inverse}\\
&(\cdot)^{-1}&\srt{S}\to\srt{S}& \mbox{Symmetric key inverse}\\
%% \mbox{Equations:}&\multicolumn{2}{l}{(x^{-1})^{-1}=x\mbox{ for
%%     $x:\srt{A}$}}&y^{-1}=y\mbox{ for $y:\srt{S}$}
\mbox{Equations:}&\multicolumn{3}{l}{(x^{-1})^{-1}=x\mbox{ for $x:\srt{A}$}}\\
 &\multicolumn{3}{l}{y^{-1}=y\mbox{ for $y:\srt{S}$}}
\end{array}$$
\caption{Simple Crypto Algebra Signature}\label{fig:signature}
\end{figure}

An order-sorted algebra is a generalization of a many-sorted algebra
in which sorts may be partially ordered~\cite{GoguenMeseguer92}.  The
carrier sets associated with ordered sorts are related by the subset
relation.

Figure~\ref{fig:signature} shows the simplification of the {\cpsa}
message algebra signature used by the examples in this paper.
Sort~$\top$ is the sort of all messages.  Messages of sort~$\srt{A}$
(asymmetric keys), sort~$\srt{S}$ (symmetric keys), and sort~$\srt{D}$
(data) are called \emph{atoms}.  Messages are generated from the atoms
using encryption $\enc{\cdot}{(\cdot)}$ and pairing $(\cdot,\cdot)$,
where the comma operation is right associative and parentheses are
omitted when the context permits.

Each variable~$x$ in an order-sorted term has a unique sort~$S$.  The
\emph{declaration} of~$x$ is $x\colon S$.  The set of variables that
occur in term~$t$ is $\vars(t)$.

The quotient term algebra generated by declarations~$X$ over the
signature in Figure~\ref{fig:signature} is written $\alga_X$.  It
is the carrier set of sort~$\top$.  The canonical representative of
each member of $\alga_X$ is the term with the fewest occurrences of
the $(\cdot)^{-1}$ operation.  Unification and matching can be
implemented in such a way that only canonical terms are
considered~\cite[Appendix~B]{cpsadesign09}.

A message~$t_0$ is \emph{carried by}~$t_1$, written $t_0\sqsubseteq
t_1$ if~$t_0$ can be derived from~$t_1$ given the right set of keys,
that is $\sqsubseteq$ is the smallest reflexive, transitive relation
such that $t_0\sqsubseteq t_0$, $t_0\sqsubseteq (t_0, t_1)$,
$t_1\sqsubseteq (t_0, t_1)$, and $t_0\sqsubseteq\enc{t_0}{t_1}$.

The use of a message algebra that is order-sorted dictates that the
logic used to express the characteristic sentence associated with a
shape analysis is also order-sorted.  Furthermore, the signature for
the logic must inherit the sorts and subsort relations from the
message algebra.

\paragraph{Implementation-Oriented Strand Spaces.}

A run of a protocol is viewed as an exchange of messages by a finite
set of local sessions of the protocol.  Each local session is called a
\index{strand}\emph{strand}~\cite{ThayerEtal99}.  The behavior of a
strand, its \index{trace}\emph{trace}, is a non-empty sequence of
messaging events.  An \index{event}\emph{event} is either a message
transmission or a reception.  Outbound message $t\in\alga_X$ is
written as~$\outbnd t$, and inbound message~$t$ is written as~$\inbnd
t$.  A message \index{originates}\emph{originates} in a trace if it is
carried by some event and the first event in which it is carried is
outbound.

\iffalse A message is \index{gained}\emph{gained} by a trace if it is
carried by some event and the first event in which it is carried is
inbound.  A message is \index{acquired}\emph{acquired} by a trace if
it first occurs in a reception event and is also carried by that
event.
\fi

A \emph{strand space}~$\Theta_X$ is a finite map from a set of strands
to their traces.  {\cpsa} represents a set of strands as an initial
segment of the natural numbers, therefore, a strand space is a
sequence of traces.  The \emph{nodes} of a strand space are
$\nodes(\Theta_X)=\{(s, i)\mid s\in\sdom(\Theta_X), 0\le
i<|\Theta_X(s)|\}$.  The event at node $n=(s,i)$ is
$\evt_\Theta(s,i)=\Theta(s)(i)$.

In a strand space, a message that originates in exactly one trace is
\emph{uniquely originating}, and represents a freshly chosen value.  A
message that originates nowhere and is never used by the adversary to
decrypt or encrypt a message is \emph{non-originating}, and represents
an uncompromised key.

A protocol~$P$ is a finite set of traces, which are the \emph{roles}
of the protocol.  Strand~$s\in\sdom(\Theta_X)$ is an
\emph{elaboration} of role~$r\in P$ if $\Theta_X(s)$ is a prefix of
the result of applying some substitution~$\sigma$ to~$r$.  An example
of a protocol is in Eq.~\ref{eq:cpsa protocol} in
Section~\ref{sec:example}.

\paragraph{Skeletons.}

A skeleton represents all or part of the regular portion of an
execution.  A skeleton contains a strand space, a partial ordering of
its nodes, assumptions about uncompromised keys and freshly generated
atoms, and role associations.

A \emph{skeleton} $k=\skel_X(\rl,P,\Theta_X,\prec,N,U)$, where
$\rl\colon\sdom(\Theta_X)\to P$ is a role map,~$\prec$ is a strict
partial ordering of the nodes,~$N$ is a set of atoms, none of which
originate in a trace in~$\Theta_X$, and~$U$ is a set of atoms, all of
which originate in no more than one trace in~$\Theta_X$.  In
addition,~$\prec$ must order the node for each event that receives a
uniquely originating atom after the node of its transmission, so as to
model the idea that the atom represents a value freshly generated when
it is transmitted.

The above definition of a skeleton is useful for defining the
semantics of shape analysis sentences, but it does not reflect the
syntax used by {\cpsa}.  In {\cpsa} syntax, the trace and the role
associated with a strand is specified by an
\index{instance}\emph{instance}.  An instance is of the form
$\insta(r,h,\sigma)$, where~$r$ is a role, $h$ specifies the length of
a trace instantiated from the role, and~$\sigma$ specifies how to
instantiate the variables in the role to obtain the trace.  Thus the
trace associated with $\insta(r,h,\sigma)$ is
$\sigma\circ\prefix{r}{h}$, the prefix of length~$h$ that results from
applying~$\sigma$ to~$r$.

In the {\cpsa} syntax, the role map and sequence of traces are
replaced by a sequence of instances.  So for skeleton
$\skel_X(\rl,P,\Theta_X,\prec,N,U)$, the {\cpsa} syntax is
$\skel_X(P,I,\prec,N,U)$, where for each $s\in\sdom(\Theta_X)$,
$I(s)=\insta(r,h,\sigma)$, $r=\rl(s)$, and the trace of
$\insta(r,h,\sigma)$ is $\Theta_X(s)$.

Two examples of skeletons are displayed in Figure~\ref{fig:blanchet's
  shape analysis} in Section~\ref{sec:example}.

\paragraph{Homomorphisms.}

Let $k_0=\skel_X(rl_0,P,\Theta_0,\prec_0,N_0,U_0)$ and
$k_1=\skel_Y(rl_1,P,\break\Theta_1,\prec_1,N_1,U_1)$ be skeletons.  There is
a \emph{skeleton homomorphism} $(\phi,\sigma)\colon k_0\mapsto k_1$
if~$\phi$ and~$\sigma$ are maps with the following properties:
\begin{enumerate}
\item $\phi$ maps strands of~$k_0$ into those of~$k_1$, and nodes as
  $\phi((s,i))=(\phi(s),i)$, that is $\phi$ is in
  $\sdom(\Theta_0)\to\sdom(\Theta_1)$;
\item $\sigma\colon\alga_X\to\alga_Y$ is a message algebra homomorphism;
\item $n\in\nodes(\Theta_0)$ implies
  $\sigma(\evt_{\Theta_0}(n))=\evt_{\Theta_1}(\phi(n))$;
\item $n_0\prec_0
n_1$ implies $\phi(n_0)\prec_1\phi(n_1)$;
\item $\sigma(N_0)\subseteq N_1$;
\item $\sigma(U_0)\subseteq U_1$;
\item\label{item:orig} $t\in U_0$ implies
  $\phi(\orig_{k_0}(t))\subseteq\orig_{k_1}(\sigma(t))$;
\end{enumerate}
where $\orig_k(t)$ is the set of nodes of events at which~$t$
originates.  Item~\ref{item:orig} says the node at which an atom
declared to be uniquely originating is preserved by homomorphisms.
Note that~$\phi$ is a strand mapping, not a node mapping as
in~\cite{Guttman09}.

\section{Shape Analysis Sentences}\label{sec:sas}

Given the definitions in the previous section, the language~$\lang(P)$
used for shape analysis sentences is quite constrained.  The signature
for terms extends the one used for the underlying message algebra with
a sort~$\srt{Z}$.  Variables of this sort denote strands.

Security goals make use of protocol specific and protocol independent
predicates.  For each role $r\in P$, there are strand length
predicates and strand parameter predicates.  There are $|r|$ unary
length predicates $P[r,h]:\srt{Z}$, with $1\leq h\leq|r|$.  Relative
to skeleton~$k$, $P[r,h](z)$ asserts that strand~$z$ in~$k$ is an
instance of~$r$ and has a length of at least~$h$.  For each
variable~$x:S$ that occurs in~$r$, there is a binary parameter
predicate $P[r,x]:\srt{Z}\times{}S$.  Relative to skeleton~$k$,
$P[r,x](z,t)$ asserts that strand~$z$ in~$k$ is an instance of~$r$ in
which~$x$ is instantiated as~$t$.

For each $B\in\{\srt{A},\srt{S},\srt{D}\}$, there are unary predicates
$\cn{non}\colon B$ and $\cn{uniq}\colon B$.  $\cn{non}(t)$ asserts~$t$
is non-originating in~$k$ and $\cn{uniq}(t)$ asserts~$t$ uniquely
originates in~$k$.

Let $m$ be the length of the longest role in~$P$.  There are $m^2$
binary precedence predicates $\cn{prec}[i,j]:\srt{Z}\times\srt{Z}$ for
$0\leq i,j< m$.  $\cn{prec}[i,j](x,y)$ asserts that node $(x,i)$ is
before node $(y,j)$ in~$k$.  There are $3m$ binary origin predicates
$\cn{orig}[i]:B\times\srt{Z}$, with $0\leq i< m$ and $B$ as before.
$\cn{orig}[i](t,z)$ asserts that~$t$ uniquely originates in~$k$ at
node $(z,i)$.  The predicate \cn{false} has arity zero and, of course,
equality is binary.

To improve the readability of formulas to follow, we write
$\cn{prec}(x,i,y,j)$ for $\cn{prec}[i,j](x,y)$ and $\cn{orig}(t,z,i)$
for $\cn{orig}[i](t,z)$.

We define~$\form(k)=(Y,\Phi)$, where~$\Phi$ is $k$'s skeleton formula,
and~$Y$ is the formula's declarations.  Using the {\cpsa} skeleton
syntax presented in Section~\ref{sec:alg and morphs}, let
$k=\skel_X(P,I,\prec,N,U)$.  The declarations~$Y$ is~$X$ augmented
with a fresh variable~$z_s\colon\srt{Z}$ for each strand
$s\in\sdom(I)$. Let~$\prec^-$ be the transitive reduction
of~$\prec$. The \emph{skeleton formula}~$\Phi$ of~$k$ is a conjunction
of atomic formulas composed as follows.

\begin{itemize}
\item For each $s\in\sdom(I)$, let $I(s)=\insta(r,h,\sigma)$.  Assert
  $P[r,h](z_s)$.  For each variable $x\in\vars(\prefix{r}{h})$ and term
  $t=\sigma(x)$, assert $P[r,x](z_s,t)$.
\item For each $(s,i)\prec^-(s',i')$ with $s\neq s'$, assert
  $\cn{prec}(z_s,i,z_{s'},i')$
\item For each $t\in N$, assert $\cn{non}(t)$.
\item For each $t\in U$, assert $\cn{uniq}(t)$.
\item For each $t\in U$ and $(s,i)\in\orig_k(t)$, assert
  $\cn{orig}(t,z_s, i)$.
\end{itemize}

Given a set of homomorphisms $\delta_i\colon k_0\mapsto k_i$, its shape
analysis sentence is
\begin{equation}\label{eqn:shape sentence}
\sent(\delta_i\colon k_0\mapsto k_i)=\all{X_0}(\Phi_0\supset
\bigvee_i\some{X_i}(\Delta_i\wedge\Phi_i)),
\end{equation}
where $\form(k_0)=(X_0,\Phi_0)$.  The same procedure produces~$X_i$
and~$\Phi_i$ for shape~$k_i$ with one proviso---the variables in
$X_i$ that also occur in~$X_0$ must be renamed to avoid trouble while
encoding the structure preserving maps~$\delta_i$.

The structure preserving maps~$\delta_i=(\phi_i,\sigma_i)$ are encoded
in~$\Delta_i$ by a conjunction of equalities.  Map~$\sigma_i$ is coded
as equalities between a message algebra variable in the domain
of~$\sigma_i$ and the term it maps to.  Map~$\phi_i$ is coded as
equalities between strand variables in~$\Phi_0$ and strand variables
in~$\Phi_i$.  Let~$Z_0$ be the sequence of strand variables freshly
generated for~$k_0$, and~$Z_i$ be the ones generated for~$k_i$.  The
strand mapping part of~$\Delta_i$ is
$\bigwedge_{j\in\sdom(\Theta)}Z_0(j)=Z_i(\phi_i(j))$.

An example shape analysis sentence is displayed in
Figure~\ref{fig:blanchet's shape analysis sentence}.  A strand length
predicate $P[r,h](z)$ is written $r_h(z)$ with the protocol left
implicit, and similarly for the strand parameter predicates.

\paragraph{Semantics of Skeleton Formulas.}

Let $k=\skel_X(\rl,P,\Theta,\prec,N,U)$.  When formula~$\Phi$ is
satisfied in skeleton~$k$ with order-sorted variable assignment
$\alpha$, we write $k,\alpha\models\Phi$.  For $x:S\in X$, $\alpha(x)$
is in the carrier set of $\alga_X$ for sort~$S$.  For $x:\srt{Z}$,
$\alpha(x)\in\sdom(\Theta)$.  We write~$\bar\alpha$ when~$\alpha$ is
extended to terms in the obvious way.  When sentence~$\Sigma$ is
satisfied in skeleton~$k$, we write $k\models\Sigma$.

\begin{itemize}
\item $k,\alpha\models P[r,h](y)$ iff $h\leq|\Theta(s)|$ and
  $\prefix{\Theta(s)}{h}=\sigma\circ\prefix{r}{h}$ for some~$\sigma$,
  where $s=\alpha(y)$.
\item $k,\alpha\models P[r,x](y, t)$ iff
  $\prefix{\Theta(s)}{h}=\sigma\circ\prefix{r}{h}$ for
  some~$\sigma$ with $\sigma(x)=\bar\alpha(t)$, where $s=\alpha(y)$
  and $h$ is the smallest $\ell$ such that $x$ occurs in
  $\prefix{r}{\ell}$.
\item $k,\alpha\models\cn{prec}(x,i,y,j)$ iff
$(\alpha(x),i)\prec(\alpha(y),j)$.
\item $k,\alpha\models\cn{non}(t)$ iff $\bar\alpha(t)\in N$.
\item $k,\alpha\models\cn{uniq}(t)$ iff $\bar\alpha(t)\in U$.
\item $k,\alpha\models\cn{orig}(t,z,i)$ iff $\bar\alpha(t)\in U$ and
  $(\alpha(z),i)\in\orig_k(\bar\alpha(t))$.
\item $k,\alpha\models y=z$ iff $\bar\alpha(y)=\bar\alpha(z)$.
\item $k,\alpha\not\models\cn{false}$.
\end{itemize}

\begin{thm}\label{thm:skeleton models}
Let $\form(k_0)=(X,\Phi)$ and $\Sigma=\some{X}\Phi$.  Sentence~$\Sigma$ is
satisfied in~$k$ iff there is a homomorphism from~$k_0$ to
$k$, i.e.\ $k\models\Sigma$ iff
$\some{\delta}\delta\colon k_0\mapsto k$.
\end{thm}

This theorem corrects the first of the two main results
from~\cite{Guttman09}, as that paper omits the \cn{orig} predicate.  A
later paper includes the \cn{orig} predicate~\cite{Guttman11a}, using
the symbol \cn{UnqAt}.

\begin{proof}
For the forward direction, assume~$\alpha$ is a variable assignment
for the variables in~$X$ such that $k,\alpha\models\Phi$, and let~$Z$
be the sequence of strand variables constructed while
generating~$\Phi$ from~$k_0$.  Then the pair of maps
$\delta=(\alpha\circ Z,\alpha)$ demonstrate a homomorphism from~$k_0$
to~$k$, i.e.\ each item in the definition of a skeleton homomorphism
in Section~\ref{sec:alg and morphs} is satisfied.

For the reverse direction, assume maps $\delta=(\phi,\sigma)$ are such
that $\delta\colon k_0\mapsto k$.  Then the desired variable assigment is
$$\alpha(x)=\left\{
\begin{array}{ll}
\phi(Z^{-1}(x))&x\in\sran(Z)\\
\sigma(x)&x\in\sdom(\sigma).
\end{array}\right.$$
\end{proof}

\paragraph{Deducing Security Goals.}  A shape analysis $\delta_i\colon
k_0\mapsto k_i$ is \emph{complete} if for each realized skeleton~$k$,
$\delta\colon k_0\mapsto k\mbox{ iff }\some{i,\delta'}\delta'\colon
k_i\mapsto k$.  There is an ongoing effort to show that whenever
{\cpsa} terminates it produces a complete shape analysis, however,
preliminary analysis suggests that with the exception of specially
constructed, artificial protocols, {\cpsa}'s output is complete.  See
Appendix~\ref{sec:artificial protocol} for an example of a troublesome
artificial protocol.

The next theorem captures the sense in which a shape analysis sentence
characterizes a complete shape analysis.

\begin{thm}\label{thm:sentence implies}
Let $\delta_i\colon k_0\mapsto k_i$ be a complete shape analysis.
Then the shape analysis sentence~$\Sigma=\sent(\delta_i\colon
k_0\mapsto k_i)$ is satisfied in all realized skeletons~$k$,
i.e.\ $k\models\Sigma$.
\end{thm}

\begin{proof}
Shapes are minimal among realized skeletons, so there is no realized
skeleton in the image of~$k$ that is not in the image of one of the
shapes.  Therefore, by Theorem~\ref{thm:skeleton models}, the negation
of the hypothesis of the implication is satisfied in all realized
skeletons that are not in the image of~$k_0$, and the disjunction is
satisfied in the remaining realized skeletons.
\end{proof}

Let~$\Sigma$ be the shape analysis sentence of a complete shape
analysis and~$\Psi$ be a security goal.  If $\Sigma\supset\Psi$ is a
theorem in order-sorted first-order logic, then~$\Psi$ is satisfied in
all realized skeletons and its protocol achieves this goal.

Since $\prec$ is transitive, transitivity of \fn{prec} can also be
used to prove a protocol achieves a goal.  That is,
\[\cn{prec}(x,i,y,j)
\land\cn{prec}(y,j,z,k)\supset
\cn{prec}(x,i,z,k).\]

\section{Detailed Example}\label{sec:example}

The simplified version of the Denning-Sacco
key distribution protocol~\cite{DenningSacco81} due to Bruno Blanchet
is now revisited.
$$\begin{array}{r@{{}:{}}l}
A\to B&\enc{\enc{s}{a^{-1}}}{b}\\
B\to A&\enc{d}{s}
\end{array}$$
Symmetric key~$s$ is freshly generated, asymmetric keys $a^{-1}$ and
$b^{-1}$ are uncompromised, and the goal of the protocol is to keep
data~$d$ secret.  This {\cpsa} description of the protocol in
Eq.~\ref{eq:protocol}, has an initiator and a responder role.
\begin{equation}\label{eq:cpsa protocol}
\begin{array}{r@{{}={}}l}
\init(a,b\colon\srt{A},s\colon\srt{S}, d\colon\srt{D})
&\seq{\outbnd\enc{\enc{s}{a^{-1}}}{b},\inbnd\enc{d}{s}}\\
\resp(a,b\colon\srt{A},s\colon\srt{S}, d\colon\srt{D})
&\seq{\inbnd\enc{\enc{s}{a^{-1}}}{b},\outbnd\enc{d}{s}}
\end{array}
\end{equation}

\begin{figure}
$$\begin{array}{@{}r@{}c@{}l@{}}
k_0&{}={}&\skel_X(\begin{array}[t]{@{}ll}
\{\init(a_0,b_0,s_0,d_0),\resp(a_1,b_1,s_1,d_1)\},
&\mbox{Protocol}\\
\seq{\insta(\resp,2,\{a_1\mapsto a,b_1\mapsto b,s_1\mapsto s,d_1\mapsto d\})},
&\mbox{Instances}\\
\emptyset,
&\mbox{Node orderings}\\
\{a^{-1},b^{-1}\},
&\mbox{Non-origination}\\
\{s\})
&\mbox{Unique origination}
\end{array}\\
&&\mbox{where $X=a,b\colon\srt{A},s\colon\srt{S}, d\colon\srt{D}$}\\
k_1&{}={}&\skel_Y(\begin{array}[t]{@{}ll}
\{\init(a_0,b_0,s_0,d_0),\resp(a_1,b_1,s_1,d_1)\},
&\mbox{Protocol}\\
\langle\begin{array}[t]{@{}l}
\insta(\resp,2,\{a_1\mapsto a,b_1\mapsto b,s_1\mapsto s,d_1\mapsto d\}),\\
\insta(\init,1,\{a_0\mapsto a,b_0\mapsto b',s_0\mapsto s\})\rangle
\end{array}
&\begin{array}[t]{@{}l}
\mbox{Instances}\\
\mbox{\emph{Note $b_0\mapsto b'$ not $b$!}}
\end{array}\\
\{(1,0)\prec(0,0)\},
&\mbox{Node orderings}\\
\{a^{-1},b^{-1}\},
&\mbox{Non-origination}\\
\{s\})
&\mbox{Unique origination}
\end{array}\\
&&\mbox{where $Y=a,b,b'\colon\srt{A},s\colon\srt{S}, d\colon\srt{D}$}\\
\delta_1&{}={}&(\seq{0},\{a\mapsto a, b\mapsto b, s\mapsto s, d\mapsto d\})
\end{array}$$
\caption{Shape Analysis for Blanchet's
  Protocol}\label{fig:blanchet's shape analysis}
\end{figure}

The protocol was constructed with a known flaw for expository
purposes, and as a result the secret is exposed due to an
authentication failure.
The desired authentication goal is:
$$\begin{array}{l}
\all{a,b\colon\srt{A}, s\colon\srt{S}, d\colon\srt{D}, z_0\colon\srt{Z}}(\\
\quad\resp_2(z_0)\wedge
\resp_a(z_0,a)\wedge
\resp_b(z_0,b)\wedge
\resp_s(z_0,s)\wedge
\resp_d(z_0,d)\wedge{}\\
\quad\cn{non}(a^{-1})\wedge
\cn{non}(b^{-1})\wedge
\cn{uniq}(s)\supset\some{z_1\colon\srt{Z}}(\init_1(z_1)\wedge\init_b(z_1,b)))
\end{array}$$
that is, when the responder~($B$) runs to completion, there is an
initiator~($A$) that is using~$b$ for the encryption of its initial
message.

To investigate this goal, we ask {\cpsa} to find out what other
regular behaviors must occur when a responder runs to completion by
giving {\cpsa} skeleton~$k_0$ in Figure~\ref{fig:blanchet's shape
  analysis}.  {\cpsa} produces shape~$k_1$ that shows that an
initiator must run, but it need not use the same key to encrypt its
first message.  The shape analysis sentence for this scenario is
displayed in Figure~\ref{fig:blanchet's shape analysis sentence}.
Needless to say, the authentication goal cannot be deduced from this
sentence due to the man-in-the-middle attack.  If one repeats the
analysis using the protocol in Eq.~\ref{eq:amended protocol}, the
generated shape analysis sentence can be used to deduce the
authentication goal.

\begin{figure}
$$\begin{array}{l}
\all{a_0,b_0\colon\srt{A}, s_0\colon\srt{S}, d_0\colon\srt{D}, z_0\colon\srt{Z}}(\\
\quad\resp_2(z_0)\wedge
\resp_a(z_0,a_0)\wedge
\resp_b(z_0,b_0)\wedge
\resp_s(z_0,s_0)\wedge
\resp_d(z_0,d_0)\wedge{}\\
\quad\cn{non}(a_0^{-1})\wedge
\cn{non}(b_0^{-1})\wedge
\cn{uniq}(s_0)\\
\quad\supset\\
\quad\some{a_1,b_1,b_2\colon\srt{A}, s_1\colon\srt{S}, d_1\colon\srt{D}, z_1,z_2\colon
 \srt{Z}}(\\
\qquad z_0=z_1\wedge a_0=a_1\wedge b_0=b_1\wedge s_0=s_1\wedge d_0=d_1\wedge
\resp_2(z_1)\wedge{}\\
\qquad\resp_a(z_1,a_1)\wedge
\resp_b(z_1,b_1)\wedge
\resp_s(z_1,s_1)\wedge
\resp_d(z_1,d_1)\wedge{}\\
\qquad\init_1(z_2)\wedge
\init_a(z_2,a_1)\wedge
\init_b(z_2,b_2)\wedge
\init_s(z_2,s_1)\wedge\cn{orig}(s_1,z_2,0)\wedge{}\\
\qquad\cn{prec}(z_2,0,z_1,0)\wedge\cn{non}(a_1^{-1})\wedge
\cn{non}(b_1^{-1})\wedge
\cn{uniq}(s_1)))
\end{array}$$
\caption{Shape Analysis Sentence for Blanchet's
  Protocol}\label{fig:blanchet's shape analysis sentence}
\end{figure}

\section{Conclusion}\label{sec:conclusion}

This paper presented a method for extracting a sentence that
completely characterizes a run of {\cpsa} and showed that logical
deduction can then be used to determine if a security goal is
satisfied.  To ensure the fidelity of the translation between {\cpsa}
output and a shape analysis sentence, an order-sorted first-order
logic is employed.  Furthermore, the first-order language used for
formulas is dictated by the {\cpsa} syntax for skeletons and
the formalization of homomorphisms used by {\cpsa}.

\section*{Acknowledgment}

Paul Rowe and Jon Millen provided valuable feedback on an early draft
of this paper.

\appendix
\section{Artificial Protocol}\label{sec:artificial protocol}

This section presents an example of a protocol that causes {\cpsa} to
fail to produce a complete shape analysis.

\begin{equation}\label{eq:artificial protocol}
\begin{array}{r@{{}={}}l}
\init(a\colon\srt{A},d\colon\srt{D})&
\seq{\outbnd\enc{d}{a},\inbnd d}\\
\resp(x\colon\top,d\colon\srt{D})&
\seq{\inbnd x, \outbnd d}
\end{array}
\end{equation}

The initiator in the
protocol specifies half of a common authentication pattern.  Assuming
nonce~$d$ is freshly generated, and key~$a^{-1}$ is uncompromised, an
execution of the protocol in which an instance of the initiator role
runs to completion must include other regular behavior by a strand
that possesses the decryption key~$a^{-1}$.

It's the responder role that is artificial.  Its first event is the
reception of a message of any sort, and then it transmits a message of
sort data.  There are many ways in which an instance of the responder
role can serve as the other half of the authentication pattern, such as:
\begin{equation}\label{eq:missing shape analysis}
%\begin{math}
\begin{array}{rcll}
%\xymatrix@R=1em@C=3em{\strut\init&&\strut\resp\\
\xymatrix@R=1em@C=3em{\raisebox{-1ex}[0ex][0ex]{\strut\init}&&
\raisebox{-1ex}[0ex][0ex]{\strut\resp}\\
\bullet\ar@{=>}[d]\ar[r]^{\enc{d}{a}}&\prec\ar[r]^{\enc{d}{a}}&\bullet\ar@{=>}[d]\\
\bullet&\ar[l]_{d}\succ&\bullet\ar[l]_{d}}
&
\raisebox{-4ex}{or}
&
\xymatrix@R=1em@C=3em{\raisebox{-1ex}[0ex][0ex]{\strut\init}&&
\raisebox{-1ex}[0ex][0ex]{\strut\resp}\\
\bullet\ar@{=>}[d]\ar[r]^{\enc{d}{a}}&\prec
\ar[r]^{(\enc{d}{a},\enc{d}{a})}&\bullet\ar@{=>}[d]\\
\bullet&\ar[l]_{d}\succ&\bullet\ar[l]_{d}}
&
\raisebox{-4ex}{$\cdots$}
\end{array}
%\end{math}
\end{equation}

Yet consider an operational interpretation of the responder strand in
Eq.~\ref{eq:missing shape analysis}.  The role states that it first
receives a message without knowing its structure, but the strand
interprets that message as something it can decrypt and extracts the
nonce.  Formalizations based on an operation semantics, such as what is
used for the Protocol Composition Logic~\cite{datta05}, exclude the
executions in Eq.~\ref{eq:missing shape analysis}, but there in
nothing in strand space theory that prohibits those executions.

\section{Node-Oriented Shape Analysis Sentences}\label{sec:node-oriented}

The language~$\bar\lang(P)$ used for node-oriented shape analysis
sentences more in tune with the language described
in~\cite{Guttman11a}.  The signature for terms extends the one used
for the underlying message algebra with a sort~$\srt{N}$, the sort for
nodes.

Security goals make use of protocol specific and protocol independent
predicates.  For each role $r\in P$ and $i<|r|$, there is a protocol
specific unary position predicate $P[r,i]:\srt{N}$.  For each
role~$r\in P$ and variable~$x:S$ that occurs in $r$, there is a
protocol specific binary parameter predicate $P[r,x]:\srt{N}\times S$.
The protocol independent unary predicates are $\mathsf{non}: B$ and
$\mathsf{uniq}: B$ for each atomic sort~$B\in\{\srt{A},
\srt{S},\srt{D}\}$.  The predicate \cn{false} has arity zero.  The
remaining protocol independent predicates are binary, and are
$\mathsf{orig}: B\times\srt{N}$,
$\mathsf{sprec}:\srt{N}\times\srt{N}$,
$\mathsf{prec}:\srt{N}\times\srt{N}$, and equality.

Soon we define~$\form(k)=(Y,\Phi)$, where~$\Phi$ is $k$'s skeleton
formula, and~$Y$ is the formula's declarations, but first we define
the relevant nodes of a skeleton~$N$.  Let
$k=\skel_X(P,I,\prec,\nu,\upsilon)$ and let~$\prec^-$ be the
transitive reduction of~$\prec$.  Recall that~$\Theta_X$ is the strand
space defined by~$I$.  The \emph{relevant nodes} of~$k$ are $N=N_s\cup
N_\prec\cup N_\upsilon$ where

\[\renewcommand{\arraystretch}{1.5}
\begin{array}{r@{{}={}}l}
  N_s&\{(s,i)\mid s\in\sdom(\Theta_X)\land
  i=|\Theta_X(s)|-1\}\\

  N_\prec&\{(s,i)\mid
  \renewcommand{\arraystretch}{1}
  \begin{array}[t]{@{}l}
    (s',i')\in\nodes(\Theta_X)\land s\neq s'\\
    \quad\land((s,i)\prec^- (s',i')\lor(s',i')\prec^-(s,i))\}
  \end{array}\\

  N_\upsilon&\{(s,i)\mid t\in\upsilon,(s,i)\in\orig_k(t)\}
\end{array}\]

For~$\form(k)=(Y,\Phi)$, the declarations~$Y$ is~$X$ augmented with a
fresh variable of sort~\srt{N} for each node in~$N$, and let $v(n)$ be
the variable associated with node~$n$.

The formula~$\Phi$ is a conjunction of atomic formulas composed as
follows.
\begin{itemize}
\item For each $(s,i)\in N$, assert $P[r,i](v(s,i))$,
  where $I(s)=\insta(r,h,\sigma$).
\item For each $s\in\sdom(I)$, let
  $I(s)=\insta(r,h,\sigma)$.  For each variable
  $x\in\vars(\prefix{r}{h})$ and term $t=\sigma(x)$, assert
  $P[r,x](v(s,h-1),t)$.
\item For each $(s,i),(s,i')\in N$ such that $i<i'$, assert
  $\cn{sprec}(v(s,i),v(s,i'))$.
\item For each $(s,i)\prec^-(s',i')$ such that $s\neq s'$, assert
  $\cn{prec}(v(s,i),v(s',i'))$.
\item For each $t\in\nu$, assert $\cn{non}(t)$.
\item For each $t\in\upsilon$, assert $\cn{uniq}(t)$.
\item For each $t\in\upsilon$ and node~$n$ such that
  $n\in\orig_k(t)$, assert $\cn{orig}(t,v(n))$.
\end{itemize}

Given a set of homomorphisms $\delta_i\colon k_0\mapsto k_i$, its shape
analysis sentence is
\begin{equation}\label{eqn:node-oriented shape sentence}
\sent(\delta_i\colon k_0\mapsto k_i)=\all{X_0}\Phi_0\supset
\bigvee_i\some{X_i}\Delta_i\wedge\Phi_i,
\end{equation}
where $\form(k_0)=(X_0,\Phi_0)$.  The same procedure produces~$X_i$
and~$\Phi_i$ for shape~$k_i$ with one proviso---the variables in
$X_i$ that also occur in~$X_0$ must be renamed to avoid trouble while
encoding the structure preserving maps~$\delta_i$.

The structure preserving maps~$\delta_i=(\phi_i,\sigma_i)$ are encoded
in~$\Delta_i$ by a conjunction of equalities.  Map~$\sigma_i$ is coded
as equalities between a message algebra variable in the domain
of~$\sigma_i$ and the term it maps to.  Map~$\phi_i$ is coded as
equalities between node variables in~$\Phi_0$ and node variables
in~$\Phi_i$.  Let~$v_0$ be the node variables freshly generated
for~$k_0$, and~$v_i$ be the ones generated for~$k_i$.  The
strand mapping part of~$\Delta_i$ is
\[\bigwedge_{(s,j)\in\sdom(v_0)}v_0(s,j)=v_i(\phi_i(s),j).\]

\paragraph{Semantics of Skeleton Formulas.}

Let $k=\skel_X(\rl,P,\Theta,\prec,N,U)$.  When formula~$\Phi$ is
satisfied in skeleton~$k$ with order-sorted variable assignment
$\alpha$, we write $k,\alpha\models\Phi$.  For $x:S\in X$, $\alpha(x)$
is in the carrier set of $\alga_X$ for sort~$S$.  For $x:\srt{N}$,
$\alpha(x)\in\nodes(\Theta)$.  We write~$\bar\alpha$ when~$\alpha$ is
extended to terms in the obvious way.  When sentence~$\Sigma$ is
satisfied in skeleton~$k$, we write $k\models\Sigma$.

\begin{itemize}
\item $k,\alpha\models P[r,i](y)$ iff
  $\alpha(y)=(s,i)$, and for some~$\sigma$,
  \[\prefix{\Theta_X(s)}{i+1}=\sigma\circ\prefix{r}{i+1}.\]

\item $k,\alpha\models P[r,x](y,t)$ iff
  $\alpha(y)=(s,i)$, $x$
  occurs in $\prefix{r}{i+1}$, and for some~$\sigma$ with
  $\sigma(x)=\bar\alpha(t)$,
  \[\prefix{\Theta_X(s)}{i+1}=\sigma\circ\prefix{r}{i+1}.\]
\end{itemize}

The interpretation of the protocol independent predicates is
straightforward.
\begin{itemize}
\item $k,\alpha\models\cn{prec}(y,z)$ iff
  $\alpha(y)\prec\alpha(z)$.
\item $k,\alpha\models\cn{sprec}(y,z)$ iff $\alpha(y)\prec\alpha(z)$,
  $\alpha(y)=(s,i)$, and $\alpha(z)=(s,i')$.
\item $k,\alpha\models\cn{non}(t)$ iff $\bar\alpha(t)\in\nu$.
\item $k,\alpha\models\cn{uniq}(y)$ iff $\bar\alpha(t)\in\upsilon$.
\item $k,\alpha\models\cn{orig}(t,y)$ iff $\bar\alpha(t)\in\upsilon$ and
  $\alpha(y)\in\orig_k(\bar\alpha(t))$.
\item $k,\alpha\models y=z$ iff $\bar\alpha(y)=\bar\alpha(z)$.
\item $k,\alpha\not\models\cn{false}$.
\end{itemize}

\begin{figure}
$$\begin{array}{l}
\all{a_0,b_0\colon\srt{A}, s_0\colon\srt{S}, d_0\colon\srt{D}, n_0\colon\srt{N}}\\
\quad\resp_1(n_0)\wedge\resp_{a}(n_0,a_0)\wedge
\resp_{b}(n_0,b_0)\\
\qquad{}\wedge
\resp_{s}(n_0,s_0)\wedge
\resp_{d}(n_0,d_0)\\
\qquad{}\wedge\cn{non}(a_0^{-1})\wedge
\cn{non}(b_0^{-1})\wedge
\cn{uniq}(s_0)\\
\quad\supset\\
\quad\some{a_1,b_1,b_2\colon\srt{A}, s_1\colon\srt{S},
  d_1\colon\srt{D}, n_1,n_2,n_3\colon
 \srt{N}}\\
\qquad n_0=n_1\wedge a_0=a_1\wedge b_0=b_1\wedge s_0=s_1\wedge
d_0=d_1\\
\qquad\quad\wedge\resp_1(n_1)\wedge
\resp_{a}(n_1,a_1)\wedge
\resp_{b}(n_1,b_1)\\
\qquad\quad\wedge
\resp_{s}(n_1,s_1)\wedge
\resp_{d}(n_1,d_1)\wedge\resp_0(n_2)\\
\qquad\quad\wedge\init_0(n_3)\wedge\init_a(n_3,a_1)\wedge
\init_b(n_3,b_2)\wedge
\init_s(n_3,s_1)\\
\qquad\quad\wedge\cn{orig}(s_1,n_2)\wedge
\cn{prec}(n_3,n_2)\wedge\cn{sprec}(n_2,n_1)\\
\qquad\quad\wedge\cn{non}(a_1^{-1})\wedge
\cn{non}(b_1^{-1})\wedge
\cn{uniq}(s_1)
\end{array}$$
\caption{Node-Oriented Shape Analysis Sentence for Blanchet's
  Protocol}\label{fig:node-oriented blanchet's shape analysis sentence}
\end{figure}

The node-oriented shape analysis sentence equivalent to the one in
Figure~\ref{fig:blanchet's shape analysis sentence} is in
Figure~\ref{fig:node-oriented blanchet's shape analysis sentence}.

Since $\prec$ is transitive, transitivity of \fn{prec} can be used to
prove a protocol achieves a goal.  That is,
\[\cn{prec}(x,y)
\land\cn{prec}(y,z)\supset
\cn{prec}(x,z).\]
Furthermore, $\cn{sprec}(x,y)\supset\cn{prec}(x,y)$.

\end{document}